\newif\ifproceedings
\proceedingsfalse
\documentclass[11pt]{llncs}
\usepackage{epsfig,amsmath,amssymb,verbatim}
\usepackage{a4wide,epsfig,amsmath,verbatim}
\usepackage{subfigure}
\usepackage{algorithm}
\usepackage{algorithmic}

\newtheorem{observation}{Observation}

\newcommand{\Reals}{{\mathbb{R}}}            
\newcommand{\eps}{\varepsilon}               

\def\C{{\cal C}}
\def\D{{\cal D}}

\def\O{{\cal O}}
\def\P{{\cal P}}

\def\S{{\cal S}}
\def\T{{\cal T}}

\begin{document}
\title{Quickest Path Queries on Transportation Network}
\author{Radwa El Shawi\inst{1} \and Joachim Gudmundsson\inst{1} \and
Christos Levcopoulos\inst{2}}

\institute{ NICTA\thanks{NICTA is funded by the Australian
Government as represented by the Department of Broadband,
Communications and the Digital Economy and the Australian Research
Council through the ICT Centre of Excellence program.}, Sydney,
Australia.
\email{\{radwa.elshawi,joachim.gudmundsson\}@nicta.com.au} \and
Department of Computer Science, Lund University.
\email{christos@cs.lth.se} } \maketitle

\begin{abstract}
This paper considers the problem of finding a quickest path
between two points in the Euclidean plane in the presence
of a transportation network. A transportation network consists
of a planar network where each road (edge) has an individual
speed. A traveller may enter and exit the network at any point
on the roads. Along any road the traveller moves with a fixed
speed depending on the road, and outside the network the
traveller moves at unit speed in any direction.

We give an exact algorithm for the basic version of the problem:
given a transportation network of total complexity $n$ in the
Euclidean plane, a source point $s$ and a destination point $t$,
find a quickest path between $s$ and $t$. We also show how the
transportation network can be preprocessed in time $O(n^2 \log n)$
into a data structure of size $O(n^2)$ such that a
$(1+\eps)$-approximate quickest path cost queries between any two
points in the plane can be answered in time $O(1/\eps^4 \log n)$.
\end{abstract}


\section{Introduction} \label{sec:Intro}

Transportation networks are a natural part of our infrastructure.
We use bus or train in our daily commute, and often walk to
connect between networks or to our final destination.

A transportation network consists of a set of $n$ non-intersecting roads, where each road has a speed. Thus a transportation network is usually modelled as a plane graph $\T(S,\C)$ in the Euclidean plane (or some other metric) whose vertices $S$ are nodes and whose edges $\C$ are roads. Furthermore,  each edge has a weight $\alpha\in (0,1]$ assigned to it. One can access or leave a road through any point on the road. In the presence of a transportation network, the distance between two points is defined to be the shortest elapsed time among all possible paths joining the two points
using the roads of the network. The induced distance, called $d_{\T}$, is called a transportation distance.

Using these notations the problem at hand is as follows:
\begin{figure} [tbh]
\begin{center}
\includegraphics[width=10cm]{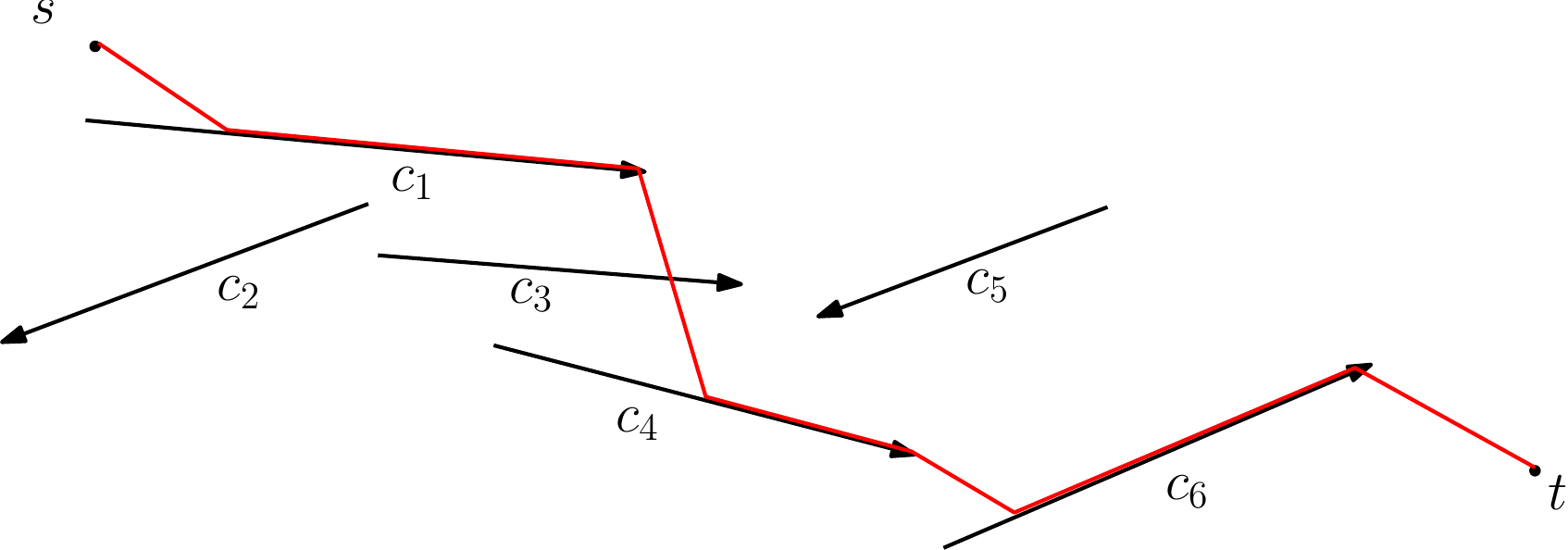}
\caption{Illustrating a quickest path from a source point $s$ to a
destination point $t$.} \label{example}
\end{center}
\end{figure}

\begin{problem} \label{prob:basic_problem}
  Given two points $s$ and $t$ in $\Reals^2$ and a transportation network $\T(S,\C)$ in the Euclidean plane. The problem is to find a path with the smallest transportation distance from $s$ to $t$, as shown in Fig.~\ref{example}.
\end{problem}

Most of the previous research has focussed on shortest paths and Voroni diagrams. Abellanas et al.~\cite{ahiklmps-pptmi-01,ahiklmprs-vdsnh-03} started work in this area considering the Voronoi diagram of a point set and shortest paths given a horizontal highway under the $L_1$-metric and the Euclidean metric. Aichholzer et al.~\cite{aap-qpssc-04} introduced the city metric induced by the $L_1$-metric and a highway network that consists of a number of axis-parallel line segments. They gave an efficient algorithm for constructing the Voronoi diagram and a quickest path map for a set of points given the city metric. G\"{o}rke et al.~\cite{gsw-ccvdf-08} and Bae et al.~\cite{bkc-occvd-09} improved and generalised these results.

In the case when the edges can have arbitrary orientation and speed, Bae et al.~\cite{GR-06} presented algorithms that compute the Voronoi diagram and shortest paths. They gave an algorithm for Problem~\ref{prob:basic_problem} that uses $O(n^3)$ time and $O(n^2)$ space. This result was recently extended to more general metrics including asymmetric convex distance functions~\cite{SPGR-05}.

In this paper we improve on the results by Bae et al.~\cite{GR-06}
and give an $O(n^2 \log n)$ time algorithm using $O(n^2)$ space.
Furthermore, we introduce the (approximate) query version. That
is, a transportation network with $n$ roads in the Euclidean plane
can be preprocessed in $O(n^2 \log n)$ time into a data  structure
of size $O(n^2/\eps^2)$ such that given any two points $s$ and $t$
in the plane a $(1+\eps)$-approximate quickest path between $s$
and $t$ can be answered in $O(1/\eps^4 \cdot \log n)$ time. For
the query structure we assume that the minimum and maximum
weights, $\alpha_{\min}$ and $\alpha_{\max}$, are constants in the
interval $(0,1]$ independent of $n$, the exact bound is stated in
Theorem~\ref{thm:main_query_result}.

This paper is organized as follows. Next we prove three
fundamental properties about an optimal path among a set of roads.
Then, in Section~\ref{sec:Buildgraph}, we show how we can use
these properties to build a graph of size $O(n^2)$ that models the
transportation network, a source point $s$, a destination point
$t$ and, contains a quickest path between $s$ and $t$. In
Section~\ref{sec:query_version} we consider the query version of
the problem. That is, preprocess the input such that an
approximate quickest path query between two query  points $s$ and
$t$ can be answered efficiently. Finally, we conclude with some
remarks and open problems.


\section{Three Properties of an Optimal Path} \label{sec:Three
Properties of an Optimal Path} In this section we are going to
prove three important properties of an optimal path. The
properties will be used repeatedly in the construction of the
algorithms in Sections ~\ref{sec:Buildgraph} and
\ref{sec:query_version}. Consider an optimal path $\P$ and let
$\C_{\P}= \langle c_1, c_2, \ldots, c_k \rangle$ be the sequence
of roads that are encountered and followed in order along $\P$,
that is, the sequence of roads on which the path changes
direction. For example, in Fig. \ref{example} that sequence would
be $\langle c_1, c_4, c_6\rangle$ but not include $c_3$ since the
path does not follow $c_3$. For any path $\P_1$, let $wt(\P_1)$
denote the cost of the path $\P_1$. Note that a road can be
encountered, and followed, several times by a path. For each road
$c_i\in \C_{\P}$ let $s_i$ and $t_i$ be the first and last point
on $c_i$ encountered for each occasion by $\P$. Without loss of
generality, we assume that $t_{i+1}$ lies below $s_{i}$ when
studying two consecutive roads in $\C_{\P}$. We have:

\begin{property} \label{Property1}
  For any two consecutive roads $c_i$ and $c_{i+1}$ in $\C_{\P}$
  the subpath of $\P$ between $t_i$ and $s_{i+1}$ is the straight
  line segment $(t_i,s_{i+1})$.
\end{property}

The endpoints of a road $c_j=(u_j,v_j)$ are denoted \textit{start
point} ($u_j$) and \textit{end point} ($v_j$), as they appear
along the road direction.

Consider a segment $(t_i,s_{i+1})$ connecting two consecutive
roads $c_i$ and $c_{i+1}$ along $\C_{\P}$. Let $\phi_{i+1}$ denote
the angle $\angle (t_i,s_{i+1},u_{i+1})$, as illustrated in
Fig.~\ref{fig:Property2}(a).

\begin{property} \label{Property2}
If $s_{i+1}$ lies in the interior of $c_{i+1}$ then
$\phi_{i+1}=\arccos(\alpha_{i+1})$.
\end{property}
\begin{proof}
 For simplicity rotate $\C$ such that $c_{i+1}$ is horizontal and lies below $t_i$, as shown in Fig.~\ref{fig:Property2}(a).
 Let $r$ denote the orthogonal projection of $t_i$ onto the ray containing $c_{i+1}$ (not necessarily on $c_{i+1}$ and let $h=|t_ir|$.
 We have:
 $$
  |t_is_{i+1}|=\frac {h}{\sin \phi_{i+1}} \quad {\rm and} \quad
  |rs_{i+1}|=h\cdot \frac {\cos \phi_{i+1}}{\sin \phi_{i+1}}.
 $$
 Thus, the cost of the path from $t_i$ to $t_{i+1}$ along $\C_{\P}$ as a function of $\phi_{i+1}$ is:
 $$
  f(\phi_{i+1}) = |t_is_{i+1}|+ \alpha_{i+1}\cdot |s_{i+1}t_{i+1}| =
            \frac {h}{\sin \phi_{i+1}} + \alpha_{i+1} \cdot (|rt_{i+1}|- h\cdot \frac {\cos \phi_{i+1}}{\sin \phi_{i+1}}).
 $$
 Differentiating the above function with respect to $\phi_{i+1}$ gives:
 $$
 f'(\phi_{i+1}) = \frac {h(\cos \phi_{i+1}-\alpha_{i+1})}{\cos^2 \phi_{i+1}-1}.
 $$
 Setting $f'(\phi_{i+1})=0$ the resulting function gives that the minimum weight path between $t_i$ and $t_{i+1}$ along $\C(\P)$ is obtained when
 $$\phi_{i+1}=\arccos(\alpha_{i+1}).$$

%
\hfill $\square$ \end{proof}

\begin{property} \label{Property3}
 There exists an optimal path $\P'$ of total cost $wt(\P)$ with $\C_{\P'}=\C_{\P}$ that fulfills Properties~\ref{Property1}-\ref{Property2}
 such that for any two consecutive roads $c_i$ and $c_{i+1}$ in $\C_{\P'}$ the straight-line segment $(t_i,s_{i+1})$ of $\P'$ must have an endpoint
 at an endpoint of $c_i$ or $c_{i+1}$, respectively.
\end{property}
\ifproceedings
 \begin{proof} The proof can be found in Appendix~A. \hfill $\square$ \end{proof}
\else
\begin{proof}
Assume the opposite, i.e., $(t_i,s_{i+1})$ does not coincide with
any of the endpoints of $c_i$ or $c_{i+1}$. Consider the three
segment path from $s_i$ to $t_{i+1}$, that is, $(s_i,t_i)$,
$(t_i,s_{i+1})$ and $(s_{i+1},t_{i+1})$. The length of this path
is:
$$\alpha_i\cdot |s_i,t_i|+ |t_i,s_{i+1}|+\alpha_{i+1}\cdot |s_{i+1},t_{i+1}|.$$
According to Lemma~\ref{Property2} the orientation of
$(t_i,s_{i+1})$ is fixed, which implies that the weight of the
path is a linear function only depending on the position of $t_i$
(or $s_{i+1}$). Hence, moving $t_i$ in one direction will
monotonically increase the weight of the path until one of two
cases occur: (1) either $t_i$ or $s_{i+1}$ encounters an endpoint
of $c_i$ or $c_{i+1}$, or (2) $t_i=s_i$ or $s_{i+1}=t_{i+1}$. If
(1) then we have a contradiction since we assumed $(t_i,s_{i+1})$
did not coincide with any endpoint. And if (2) then we have a
contradiction since $\P$ must follow both $c_i$ and $c_{i+1}$ (again from the definition of $\C_{\P}$).
\hfill $\square$ \end{proof}
\fi

\begin{figure} [tbh]
\begin{center}
\includegraphics[width=13cm]{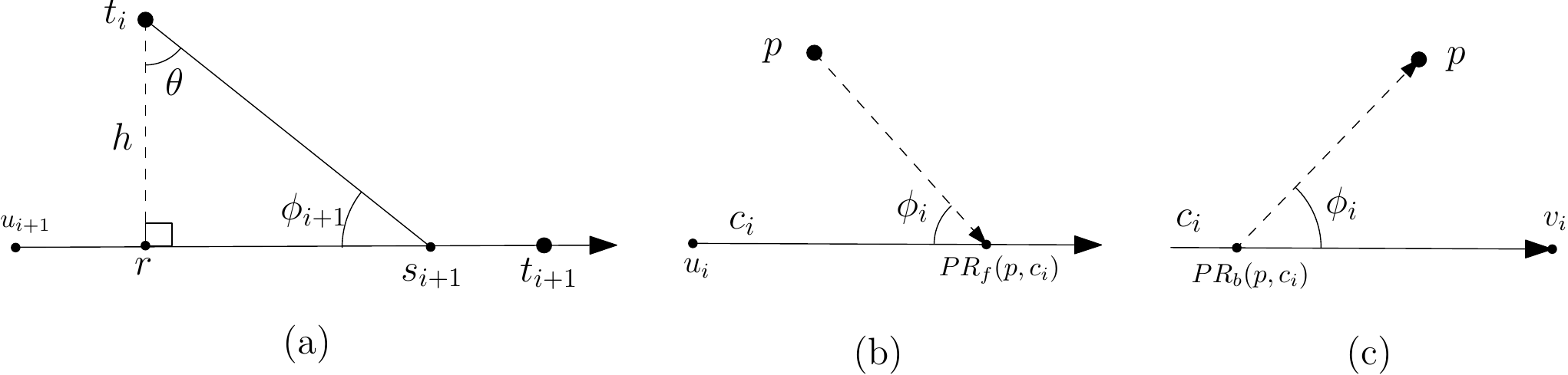}
\caption{(a) Illustrating Property~\ref{Property2}. (b) Defining $PR_f(s,c_i)$ and
(c) $PR_b(t,c_i)$.} \label{fig:Property2}
\end{center}
\end{figure}


\section{The basic case}\label{sec:Buildgraph}

In this section we consider Problem~\ref{prob:basic_problem}, that
is, as input we are given a source point $s$, a destination point
$t$ and a transportation network $\T(S,\C)$ and the aim is to find
a path with minimum transportation distance from $s$ to~$t$. Our
algorithm will construct a graph $G$ that models the set $\C$ of
roads and quickest paths between the roads. Using the three
properties shown in the previous section we will show that if a
shortest path in $G$ between $s$ and~$t$ has cost $W$ then an
optimal path between $s$ and $t$ has cost $W$. The optimal path
can then be found by running Dijkstra's
algorithm~\cite{d-ntpcg-59} on $G$.

Fix an optimal path $\P$ that fulfills
Properties~\ref{Property1}-\ref{Property3} (we know such a path
exists). If $\P$ follows~$c_i$ and $c_{i+1}$ then the path between
$c_i$ and $c_{i+1}$ (a) is a straight line segment, (b) the
segment $(t_i,s_{i+1})$ forms a fixed angle with $c_{i+1}$, and
(c) at least one of its endpoints coincides with an endpoint of
$c_i$ or $c_{i+1}$. These three properties suggest that $\P$ has a
very restricted structure which we will try to take advantage of.

Let $PR_f(p,c_i)$ be the projection of a point $p$ onto a road
$c_i$, if it exists, such that the angle
$\angle(p,PR_f(p,c_i),u_i)$ is $\phi_i$, as shown in
Fig.~\ref{fig:Property2}(b). Furthermore, let $PR_b(p,c_i)$ be the
projection of point $p$ on a road $c_i$, if it exists, such that
the angle $\angle(p,PR_b(p,c_i),v_i)$ is $\phi_i$, as shown in
Fig.~\ref{fig:Property2}(c).

Consider the graph $G(V,E)$ with vertex set $V$ and edge set $E$.
Initially $V$ and $E$ are empty. The graph $G$ is defined as
follows:

\begin{enumerate}
  \item Add $s$ and $t$ as vertices to $V$.
  \item Add the nodes in $S$ as vertices to $V$.
  \item For every road $c_i \in \C$ add the point $PR_f(s,c_i)$ (if it exists) as a vertex to $V$ and add the directed edge $(s, PR_f(s,c_i))$ (if it exists) with weight $|sPR_f(s,c_i)|$ to $E$, see Fig.~\ref{fig:BuildGraph}(a).
  \item For every road $c_i \in \C$ add the point $PR_b(t,c_i)$ (if it exists) as a vertex to $V$ and add the directed edge $(PR_b(s,c_i),t)$ (if it exists) with weight $|PR_b(s,c_i)t|$ to $E$.
  \item For every pair of roads $c_i, c_j \in \C$ add the following points (if they exist) as vertices to $V$: $PR_f(v_i, c_j)$, $PR_f(u_i, c_j)$, $PR_f(u_j,c_i)$ and $PR_f(v_j,c_i)$. Add the following four directed edges to $E$ (if their endpoints exist): $(v_i, PR_f(v_i,c_j))$, $(u_i, PR_f(u_i,c_j))$, $(v_j, PR_f(v_j,c_i))$ and $(u_j,PR_f(u_j,c_i))$. The weight of an edge is equal to the Euclidean distance between its endpoints, see Fig.~\ref{fig:BuildGraph}(b).
  \item For every pair of roads $c_i, c_j \in \C$ add the directed edges $(v_i,u_j)$ with weight $|v_iu_j|$ and $(v_j,u_i)$ with weight $|v_ju_i|$ to $E$.
  \item For every road $c_i$ consider the vertices of $V$ that correspond to points on $c_i$ in order from $u_i$ to $v_i$.
  For every consecutive pair of vertices $x_j,x_{j+1}$ along $c_i$ add a directed edge from $x_j$ to $x_{j+1}$ of weight
  $\alpha_i\cdot |x_jx_{j+1}|$, as shown in Fig.~\ref{fig:BuildGraph}(c).
\end{enumerate}

\begin{figure} [tbh]
\begin{center}
\includegraphics[width=13cm]{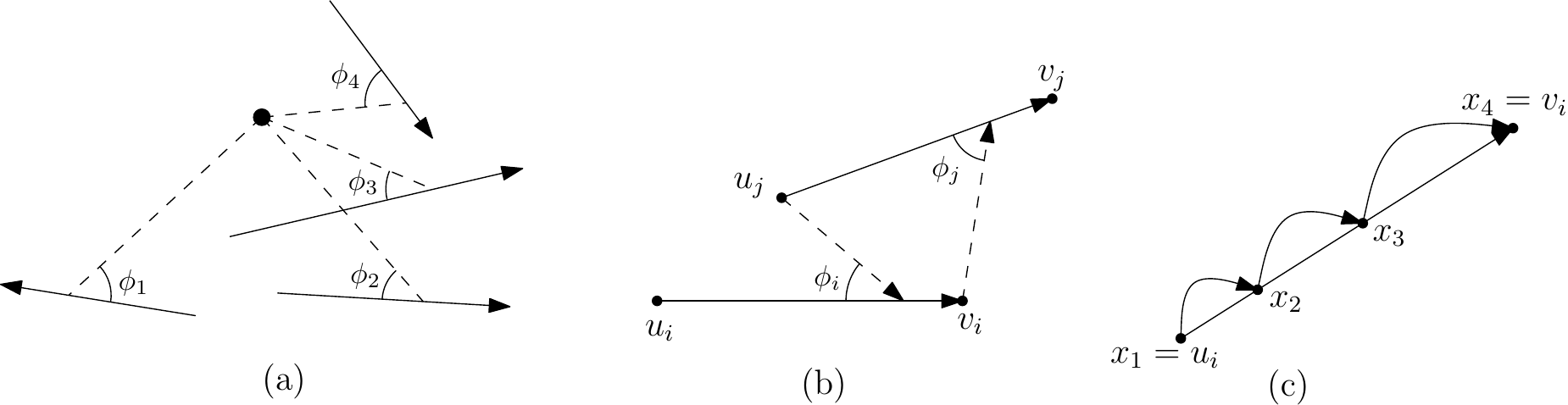}
\caption{Illustrating how the graph is built.}
\label{fig:BuildGraph}
\end{center}
\end{figure}

\begin{lemma} \label{lemma:building_graph}
  The graph $G$ contains $O(n^2)$ vertices and $O(n^2)$ edges and can be constructed in time $O(n^2 \log n)$.
\end{lemma}
\begin{proof}
  For every pair of roads we construct a constant number of vertices
  and edges that are added to $V$ and $E$, thus $O(n^2)$
  vertices and edges in total. For the first five steps of the
  construction the time to construct the vertices and edges is
  linear with respect to the size of the graph, since every edge
  and vertex can be constructed in constant time. In step~6 we need
  to sort $O(n)$ vertices along each road, thus $O(n^2 \log n)$ time
  in total.\hfill $\square$
\end{proof}

The following observation follows immediately from the construction of the graph and Properties~\ref{Property1}-\ref{Property3}.

\begin{observation} \label{obs:shortestpath}
  The shortest path between $s$ and $t$ in $G$ has cost $W$ if and only if the minimum transportation distance from $s$ to $t$ has cost $W$.
\end{observation}

By simply running Dijkstra's algorithm~\cite{d-ntpcg-59}, implemented using Fibonacci heaps, on $G$ gives the main result of this section.
\begin{theorem}
  A path with minimum transportation distance between $s$ and $t$ can be computed in $O(n^2 \log n)$ time using $O(n^2)$ space.
\end{theorem}

\subsection{Shortest paths among polygon obstacles} \label{ssec:obstacles}
In this section we briefly discuss how the above algorithm can be
generalised to the case when the plane contains polygonal obstacles.
As input we are given a source point $s$, a destination point $t$,
a transportation network $\T(S,\C)$ in the Euclidean plane and a set
$\O$ of $k$ non-intersecting obstacles of total complexity $m$.

Every edge of an obstacle can be viewed as an undirected road (or
two directed edges) with associated cost function $\alpha=1$.
Consequently, the edge connecting a road and an obstacle edge
along the optimal path has the three properties described in
Section~\ref{sec:Three Properties of an Optimal Path}. However,
while constructing the graph we have to add one additional
constraint, namely, no edge in $E$ can intersect an obstacle.
According to these properties we are going to build the graph $G$
that models the set of roads and obstacles.

There are several methods to check if a segment intersects an
obstacle~\cite{as-anspt-93,cj-arsis-92,hs-parss-95}. We will use
the data structure by Agarwal and Sharir~\cite{as-anspt-93} which
has $O(m^{1+\eps}/\sqrt L)$ query time using $O(L^{1+\eps})$
preprocessing and space for $m\leq L \leq m^2$. Using this
structure with $L=m^2$ gives us the following results:

\begin{lemma} \label{lemma:building_graph2}
 The graph $G$ has size $O(N^{2+\eps})$ and can be constructed in time $O(N^{2+\eps})$, where $N=n+m$.
\end{lemma}

By simply running Dijkstra's algorithm, implemented using Fibonacci heaps, on $G$ gives:

\begin{theorem}
 A collision-free path with minimum transportation distance between $s$ and $t$ among $\O$ can be computed in $O(N^2 \log N)$ time using $O(N^2)$ space, where $N=n+m$.
\end{theorem}

\section{Shortest Path Queries} \label{sec:query_version}

In this section we turn our attention to the query version. That
is, preprocess $\C$ such that given any two points $s$ and $t$ in
$\Reals^2$ find a quickest path between $s$ and $t$ among $\C$
effectively. We will present a data structure $\D$ that returns an
approximate quickest path. That is, given two query points $s$ and
$t$, and a positive real value $\eps$, the data structure returns
a path between~$s$ and $t$ having transportation distance at most
$(1+\eps)$ times the cost of an optimal path between~$s$ and $t$.

To simplify the description we will start (Sections~\ref{ssec:constructing_candidates}-\ref{ssec:preprocessquery}) with the case when $t$ is already known in advance and we are only given the start point $s$ as a query. Then in Section~\ref{sec:generalization} we generalize this to the case when both $s$ and $t$ are given as a query, and in Section~\ref{ssec:wspd} we show how one can improve the preprocessing time and space requirement using the well-separated pair decomposition.

Fix an optimal path $\P$ that fulfills
Properties~\ref{Property1}-\ref{Property3} (we know such a path
exists), and consider the first segment $\ell=(s,s_1)$ of $\P$.
Obviously $s_1$ must be either a start/endpoint of a road (type 1)
or an interior point of a road $c_i\in \C$ (type 2) such that
$\ell$ and $c_i$ form an angle of $\phi_i$ (ignoring the trivial
case when $s_1=t$). We will use this observation to develop an
approximation algorithm. The idea is simple. Build a graph
$G(V,E)$, as described in Section~\ref{sec:Buildgraph}, with $\T$
and $t$ as input. Compute the cost of the quickest path between
$t$ and every vertex in $V$. Now, given a query point $s$, find a
good vertex $s_1$ in $V$ to connect $s$ (either directly or via a
$2$-link path) to and then lookup the cost of the quickest path
from $s_1$ to $t$ in $G$. Obviously the problem is how to find a
``good" vertex. In the next subsection we will select a constant
number of candidate vertices such that we can guarantee that at
least one of the vertices will be a ``good" candidate, i.e., there
exists a path, that fulfills
Properties~\ref{Property1}-\ref{Property3}, from $s$ to $t$ via
this vertex that has cost at most $(1+\eps)$ times the cost of an
optimal path.


\subsection{Finding good candidates: Type 1 and Type 2} \label{ssec:constructing_candidates}
Let $\S_{\C}$ denote the set of the endpoints (both start and end)
of the roads in $\C$ and let $s$ be the query point. As described
above we will have two types to consider, and we will construct a
set $\D_1$ for the type 1 cases and a set $\D_2$ for the type 2
cases. The first set, $\D_1$, is a subset of $\S_{\C}$ and the
second set, $\D_2$, is a set of 3-tuples that will be used by the
query process (described in Section~\ref{ssec:preprocessquery}) to
calculate the quickest path.


\subsubsection{Type 1:} For the point set $\D_1$ we will use the same idea as is used in the construction of $\theta$-graphs~\cite{ns-gsn-07}. Partition the plane into a set of $k=\max\{9,\frac {36\pi}{\eps}\}$ equal size cones, denoted $X_1, \ldots, X_k$, with apex at $s$ and spanning an angle of $\theta=2\pi/k$, as shown in Fig.~\ref{fig:Query1}a. For each cone $X$ the set $\D_1$ contains a point $r$, where $r$ is a point in $\S_{\C}\cap X$ whose orthogonal projection onto the bisector of $X$ is closest to $s$. The following holds:

\begin{lemma} \label{lem:type1_set}
 Given a point set $\S_{\C}$ and a positive constant $\eps$ one can preprocess $\S_{\C}$ into a data structure of size $O(n/\eps)$ in $O(1/\eps \cdot n \log n)$ time such that given a query point $s$ the point set $\D_1$, of size at most $36\pi/\eps$, can be reported in $O(1/\eps \cdot \log n)$ time.
\end{lemma}
\ifproceedings
 \begin{proof} The proof can be found in Appendix~A. \hfill $\square$ \end{proof}
\else
\begin{proof}
  Given a direction $d$ and a point $s$ let $\ell_d(s)$ denote the infinite ray originating at $s$ with direction~$d$, see Fig.~\ref{fig:Query1}b. Let $C(s,d,\theta)$ be the cone with apex at $s$, bisector $\ell_d(s)$ and angle $\theta$. It has been shown (see for example Section~4.1.2 in~\cite{ns-gsn-07} or Lemma~2 in~\cite{bgm-otg-04}) that $\S_{\C}$ can be preprocessed in $O(n \log n)$ time into a data structure of size $O(n)$ such that given a query point $s$ in the plane the data structure returns the point in $C(s,d,\theta)$ whose orthogonal projection onto $\ell_d(s)$ is closest to $s$ in $O(\log n)$ time. We have $36\pi/\eps$ directions, thus the lemma follows. \hfill $\square$
\end{proof}
\fi

 \begin{figure} [h]
   \begin{center}
      \includegraphics[width=10cm]{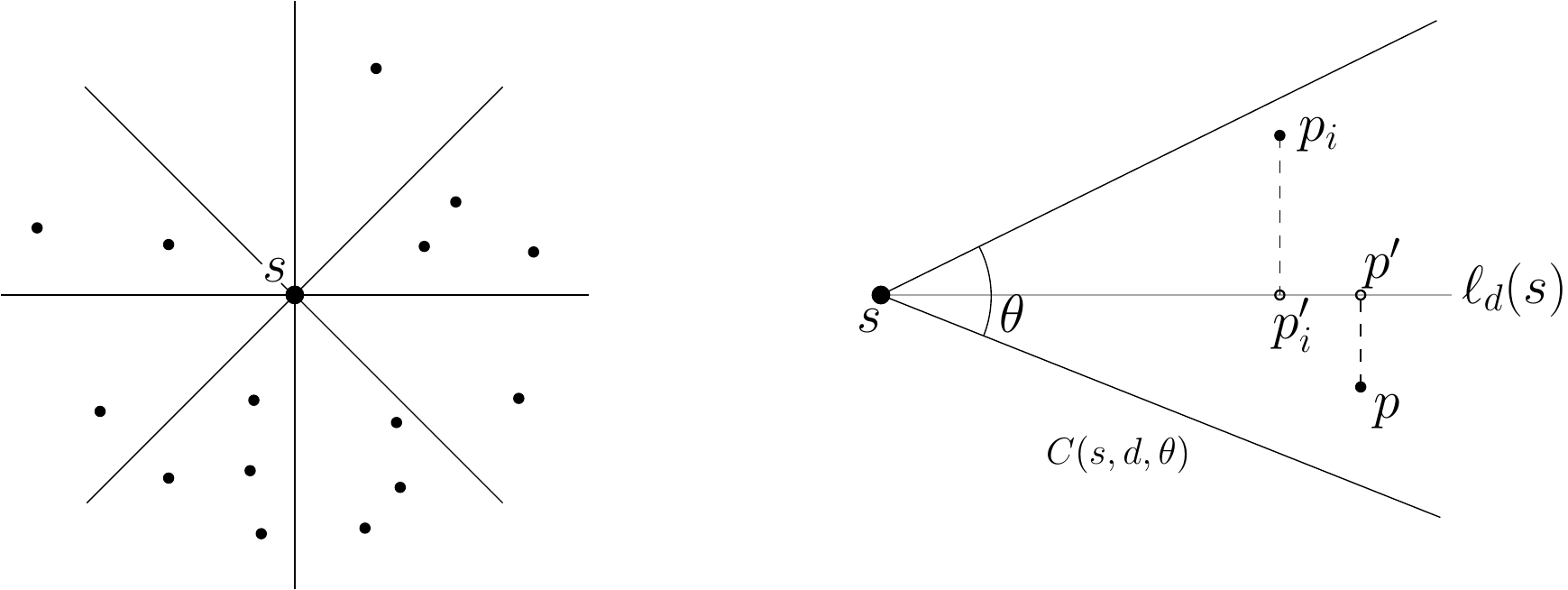}
      \caption{(a) Partitioning the plane into $k$ cones. (b) Selecting the point whose orthogonal projection onto the bisector of $X$.} \label{fig:Query1}
   \end{center}
 \end{figure}


\subsubsection{Type 2:} It remains to construct the set $\D_2$ of $3$-tuples. Unfortunately the construction might look
unnecessarily complicated but hopefully it will become clear, when
we prove the approximation bound
(Section~\ref{ssec:approximation_bound}) why we need this
construction. Before constructing $\D_2$ we need some basic
definitions.

Let $\alpha_{\max}$ and $\alpha_{\min}$ be the maximum and minimum weight of the roads in $\C$. Partition the set $\C$ of roads into a minimum number of sets, $\C_1, \ldots , \C_m$, such that the orientation of a road in $\C_i$ is in $[(i-1)\theta \cdot \alpha_{\min}, i\theta \cdot \alpha_{\min})$, where $\theta=\eps/18$. Partition each set $\C_i$, $1\leq i \leq m$, into $b$ sets $\C_{i,1}, \ldots , \C_{i,b}$ such that the weight of a road in $\C_{i,j}$ is in $[\alpha_{\min} \cdot (1+\eps)^{j-1},\alpha_{\min} \cdot (1+\eps)^j)$.

For every $i,j$, $1\leq i \leq m$ and $1\leq j \leq b$, define two
directions $\gamma_{up}(i,j)$ and $\gamma_{down}(i,j)$ as follows
(see also Fig.~\ref{fig:Query2}a). Consider an infinite ray
$\Upsilon$ with orientation $(i-1)\theta$ having weight
$\alpha=\alpha_{\min} \cdot \eps^{j-1}$. For simplicity we rotate
the ray such that it is horizontal directed from left to right.
Let $\gamma_{up}(i,j)$ be the direction of a ray $r$ originating
from below $\Upsilon$ such that $r$ and $\Upsilon$ meet at an
angle of $\arccos(\alpha)$. The direction $\gamma_{down}(i,j)$ is
defined symmetrically but with a ray originating from above
$\Upsilon$.

Given a point $p$ on a road $c$ let $N(p,c)$ denote the nearest vertex of $G$ to $p$ along $c$. Note that $N(p,c)$ must lie between $p$ and the end point of $c$.

Now, we are ready to construct $\D_2$. When given the query point $s$ construct the set $\D_2$ as follows. For each $i,j$, $1\leq i \leq m$ and $1\leq j \leq b$, shoot a ray $r_{up}$ originating from $s$ in direction $\gamma_{up}(i,j)$ and one ray $r_{down}$ in direction $\gamma_{down}(i,j)$. If $r_{up}$ hits a road in $\C_{i,j}$ then let $c_{up}$ be the first road hit and let $p_{up}$ be the point hit on $c_{up}$, as illustrated in Fig.~\ref{fig:Query2}b. The $3$-tuple $[c_{up},p_{up}, N(p_{up},c_{up})]$ is added to $\D_2$. If $r_{down}$ hits a road in $\C_{i,j}$ then let $c_{down}$ be the first road hit and let $p_{down}$ be the point hit on $c_{down}$. The $3$-tuple $[c_{down},p_{down}, N(p_{down},c_{down})]$ is added
to $\D_2$.

\begin{lemma} \label{lem:type2_set}
 Given a transportation network $\T(S,\C)$ with $n$ roads in the Euclidean plane and a positive constant $\eps$ one can preprocess $\T$ in $O(n \log n)$ time into a data structure of size $O(n)$ such that given a query point $s$ the set $\D_2$ can be reported in $O(\frac {1}{\alpha_{\min}\cdot\eps^2} \cdot \log n \log_{1+\eps} \frac{\alpha_{\max}}{\alpha_{\min}})$ time. The number of $3$-tuples in $\D_2$ is $O(\frac{1} {\alpha_{\min} \eps^2} \cdot \log_{1+\eps} \frac{\alpha_{\max}}{\alpha_{\min}})$.
\end{lemma}
\ifproceedings
 \begin{proof} The proof can be found in Appendix~A. \hfill $\square$ \end{proof}
\else
\begin{proof}
The preprocessing consists of two steps: (1) partitioning $\C$ into the sets $\C_{i,j}$, $1\leq i \leq m$ and $1\leq j \leq b$, and (2) preprocessing each set $\C_{i,j}$ into a data structure that answers ray shooting queries efficiently.

The first part is easily done in $O(n \log n)$ time by sorting the roads first with respect to their orientation and then with respect to their weight.

The second step of the preprocessing can be done by building two trapezoidal maps $T_{up}(\C_{i,j})$ and $T_{down}(\C_{i,j})$ (also know as a vertical decomposition) of each set $\C_{i,j}$ as follows (see Fig.~\ref{fig:Query2}c). Rotate $\C_{i,j}$ such that $\gamma_{up}(i,j)$ is vertical and upward. Build a trapezoidal map $T_{up}(\C_{i,j})$ of $\C_{i,j}$ as described in Chapter~6.1 in~\cite{bcko-cgaa-08}. Then preprocess $T_{up}(\C_{i,j})$ to allow for planar point location. Note that every face in $T(\C)$ either is a triangle or a trapezoid, and the left and right edges of each face (if they exists) are vertical. The trapezoidal map $T_{down}(\C_{i,j})$ can be computed in the same way by rotating $\C_{i,j}$ such that $\gamma_{down}(i,j)$ is vertical and upward. The total time needed for this step is $O(n \log n)$ and it requires $O(n)$ space.

When a query point $s$ is given, two ray shooting queries are performed for each set $\C_{i,j}$. However, instead we perform a point location in the trapezoidal maps $T_{up}(\C_{i,j})$ and $T_{down}(\C_{i,j})$. Consider $T_{up}(\C_{i,j})$ and let $f$ be the face in the map containing $s$. The top edge of $f$ corresponds to the first road $c_{up}$ hit by a ray emanating from $s$ in direction $\gamma_{up}(i,j)$. When $c_{up}$ is found we just add to $\D_2$ the first vertex on $c_{up}$ in $G$ to the right of $s$. The same process is repeated for $T_{down}(\C_{i,j})$. Performing the point location requires $O(\log n)$ time per trapezoidal map, thus the total query time is $O(mb \log n) = O(\frac {1}{\alpha_{\min}\cdot\eps^2} \cdot \log n \log_{1+\eps} (\alpha_{\max}/\alpha_{\min}))$. \hfill $\square$
\end{proof}
\fi

 \begin{figure} [h]
   \begin{center}
      \includegraphics[width=14cm]{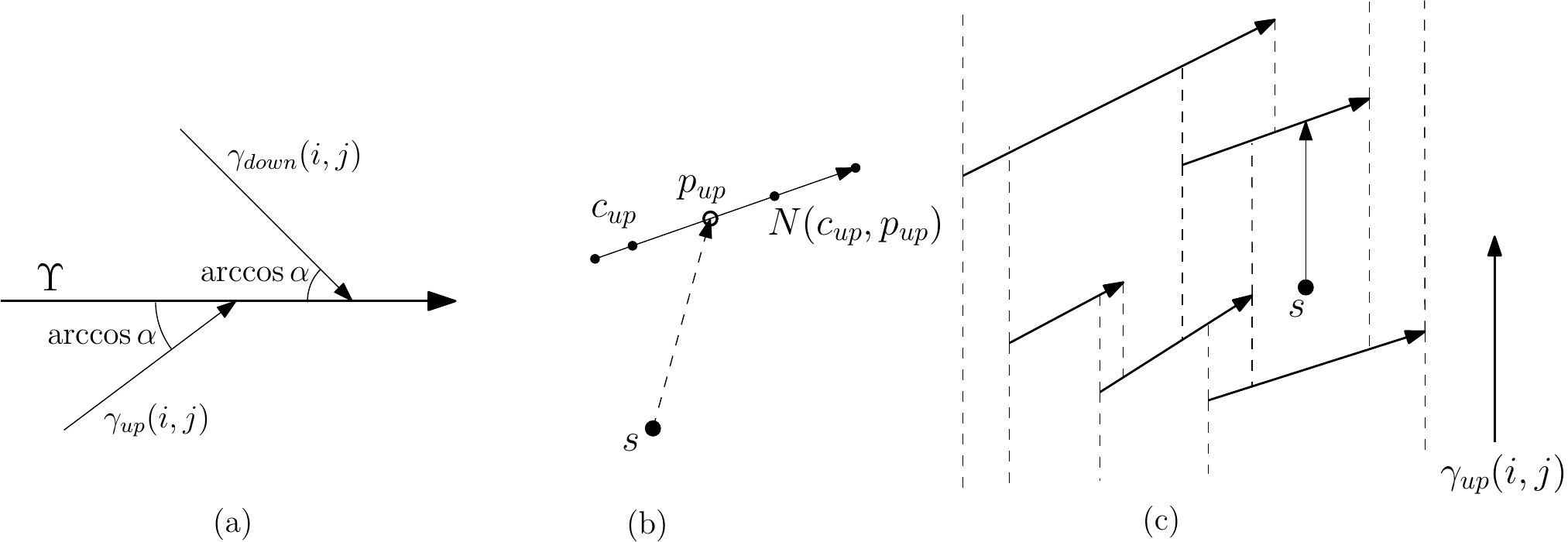}
      \caption{(a) Illustrating the definition of $\gamma_{up}(i,j)$ and $\gamma_{down}(i,j)$. (b) $p_{up}$ is the first point hit by the ray and $N(c_{up},p_{up})$ is the nearest neighbour of $p_{up}$ along $c_{up}$. (c) The trapezoidal map of a set $\C_{i,j}$ and the query point~$s$.} \label{fig:Query2}
   \end{center}
 \end{figure}


\subsection{The preprocessing and the query}\label{ssec:preprocessquery}

In this section we will present the remaining data structure, define the preprocessing and show how a query is processed.

\subsubsection{Preprocessing}
In Section~\ref{sec:Buildgraph} we showed how to build a graph given two points and the transportation network $\T(S,\C)$. The first step of the preprocessing is to build a graph $G(V,E)$ of $\C$ and the destination point $t$ (without the source point). Next compute the shortest path from every vertex in $G$ to the vertex corresponding to $t$. Since the complexity of $G$ is quadratic, this step can be done in $O(n^2 \log n)$ time and $O(n^2)$ space using Dijkstra's algorithm\footnote{The SSSP has to be performed from $t$ on a graph $G'$ where every edge has swapped direction.}. The distances are saved in a matrix $M$. Finally, we combine the above results with Lemmas~\ref{lem:type1_set} and~\ref{lem:type2_set} and get:

\begin{theorem} \label{thm:preprocessing}
 The preprocessing requires $O(n^2 \log n)$ time and $O(n^2)$ space.
\end{theorem}

\subsubsection{Query} \label{ssec:query_complexity}
As a query we are given a point $s$ in the plane. First we compute
the two sets $\D_1$ and $\D_2$.  For each vertex $s_1$ in $\D_1$
compute the quickest path via $s_1$, that is, $|ss_1|+M[s_1,t]$.
The quickest among these paths is denoted $\P_1$. The $3$-tuples
in $\D_2$ require slightly more computation. For each $3$-tuple
$[c,s_1,s_2]$ consider the path from $s$ to $t$ using $(s,s_1)$
and $(s_1,s_2)$. The cost of the path can be calculated as
$|ss_1|+\alpha \cdot |s_1s_2|+M[s_2,t]$, where $\alpha$ is the
weight of road $c$. The quickest among these paths is denoted
$\P_2$.

The quickest path among $\P_1$, $\P_2$ and the direct path from
$s$ to $t$ is then reported.

\begin{theorem} \label{thm:query}
 A query can be answered in time $O(\frac{1}{\alpha_{\min}\eps^2} \cdot \log n \log_{1+\eps} \frac{\alpha_{\max}}{\alpha_{\min}})$.
\end{theorem}
\ifproceedings
 \begin{proof} The proof can be found in Appendix~A. \hfill $\square$ \end{proof}
\else
\begin{proof}
 As above we divide the analysis into two parts: type 1 and type 2.
 \begin{description}
  \item[] Type 1: According to Lemma~\ref{lem:type1_set} the number of type 1 candidate points is at most $32\pi/\eps$ and can be computed in time $O(1/\eps\cdot \log n)$. Computing the cost of a quickest path for a point in $\D_1$ can be done in constant time. Thus, the query time is $O(1/\eps \log n)$.
  \item[] Type 2: According to Lemma~\ref{lem:type2_set} the number of $3$-tuples in $\D_2$ is $O(\frac{1}{\alpha_{\min}\eps^2} \cdot \log_{1+\eps} \alpha_{\max}/\alpha_{\min})$ and can be computed in time $O(\frac{1}{\alpha_{\min}\eps^2} \cdot \log n \log_{1+\eps} \alpha_{\max}/\alpha_{\min})$. Each element in $\D_2$ is then processed in $O(1)$ time, thus $O(\frac{1}{\alpha_{\min}\eps^2} \cdot \log n \log_{1+\eps} \alpha_{\max}/\alpha_{\min})$ in total.
 \end{description}
 Summing up we get the bound stated in the theorem.
\hfill $\square$ \end{proof}
\fi


\subsection{Approximation bound} \label{ssec:approximation_bound}

Consider an optimal path $\P$ and let $\ell=(s,s_1)$ be the first segment of $\P$. According to Property~\ref{Property1} it is a straight-line segment. For any two points $p_1$ and $p_2$ along $\P$, let $\delta_G(p_1,p_2)$ denote the cost of $\P$ from $p_1$ to $p_2$.  Define $s(\ell)$ to be the sector with apex at $s$, bisector on $\ell$, interior angle $\kappa=\eps/18$ and radius $(1+4\kappa)\cdot |ss_1|$, as shown in Fig.~\ref{fig:apx_bound1}(a).

\begin{lemma} \label{lem:apx_bound1}
 Let $\eps$ be a positive constant. If $s(\ell)\cap \S_{\C}\neq \emptyset$ then:
    $$wt(\P) \leq wt(\P_1) \leq (1+\eps) \cdot wt(\P).$$
\end{lemma}
\begin{proof}
 The proof will be shown in two steps: (1) first we prove that for every endpoint $p$ within $s(\ell)$ the quickest path, denoted
 $\P(p)$, from $s$ to $t$ using the segment $(s,p)$ has cost at most $(1+12\kappa)\cdot wt(\P)$, and then (2) we prove that $wt(\P_1)$
 has cost at most $(1+3\kappa) \cdot wt(\P(p))$. Combining the two parts proves the lemma.

 Part 1: Consider any point $p$ within $s(\ell)$, and let $\P(p)$ denote the quickest path from $s$ to $t$ using the segment $(s,p)$. We have:
 \begin{eqnarray*}
  wt(\P(p)) & = & |sp|+ \delta_G(p,t) \\
          & \leq & |sp| + |ps_1| +\delta_G(s_1,t) \\
          & \leq & (1+4\kappa) \cdot |ss_1| + \sqrt{(\cos \frac {\kappa} 2-1-4\kappa)^2+\sin^2 \frac {\kappa} 2}\cdot |ss_1| + \delta_G(s_1,t)\\
          & < & (1+12\kappa) \cdot |ss_1| + \delta_G(s_1,t)\\
          & \leq & (1+12\kappa)\cdot \delta_G(s,t)
 \end{eqnarray*}

 Part 2: As above let $p$ be any endpoint of $\C$ within $s(\ell)$, and let $\P(p)$ be the quickest path from $s$ to $t$
 using the segment $|sp|$.

 Consider the set $\D_1$ as described in Section~\ref{ssec:constructing_candidates} and assume without loss of generality that $p$ lies in a cone $X$. If there exists a point $q$ in $\D_1$ such that $q=p$ then we are done since $|sq| + \delta_G(q,t) = wt(\P(p))$.

 Otherwise, there must exist another point $q$ in $\D_1$ such that $q \in X$ and whose orthogonal projection $q'$ onto the bisector of $X$ is
 closer to $s$ than the orthogonal projection $p'$ of $p$ onto the bisector of $X$.
 \begin{eqnarray*}
    wt(\P(q)) & \leq & |sq|+|qp|+\delta_G(p,t) \\
             & \leq & |sq| + |qq'|+ |q'p'|+|p'p| +\delta_G(p,t) \\
             & \leq & |sq'|(\frac {1}{\cos \kappa}+\frac {\sin \kappa}{\cos \kappa})+|q'p'|+|sp|\cdot \sin \kappa +\delta_G(p,t)\\
             &  <   & (|sq'|+|q'p'|)(\frac {1}{\cos \kappa}+\frac {\sin \kappa}{\cos \kappa}) +|sp|\cdot \sin \kappa +\delta_G(p,t)\\
             & \leq & |sp| (\frac {1}{\cos \kappa}+\frac {\sin \kappa}{\cos \kappa}+ \sin \kappa)+\delta_G(p,t)\\
             &  <   & \frac {|sp|}{\cos \kappa} (1+ 2\sin \kappa)+\delta_G(p,t)\\
                    &  <   & |sp| \cdot (1+3\cot \kappa)+\delta_G(p,t)\\
                    &  <   & (1+3\kappa) \cdot wt(\P(p))
 \end{eqnarray*}
 In the last step we used that $\kappa=\eps/18<2\pi/9$.

 Now we can combine the two results as follows.
 \begin{eqnarray*}
 wt(\P_1) & \leq & wt(\P(q)) \\
          & \leq & (1+12\kappa) \cdot \delta_G(p,t)  \hspace{2.65cm} \mathrm{(from~Part~1)}\\
          & \leq & (1+3\kappa)\cdot (1+12\kappa) \cdot \delta_G(s,t) \hspace{1cm} \mathrm{(from~Part~2)} \\
          & < &  (1+\eps) \cdot \delta_G(s,t) \hspace{3.1cm} \mathrm{(since~}\kappa=\eps/18\mathrm{~and~}\eps<1)
 \end{eqnarray*}
 This completes the proof of the lemma.
\hfill $\square$ \end{proof}

\begin{figure} [h]
   \begin{center}
      \includegraphics[width=11cm]{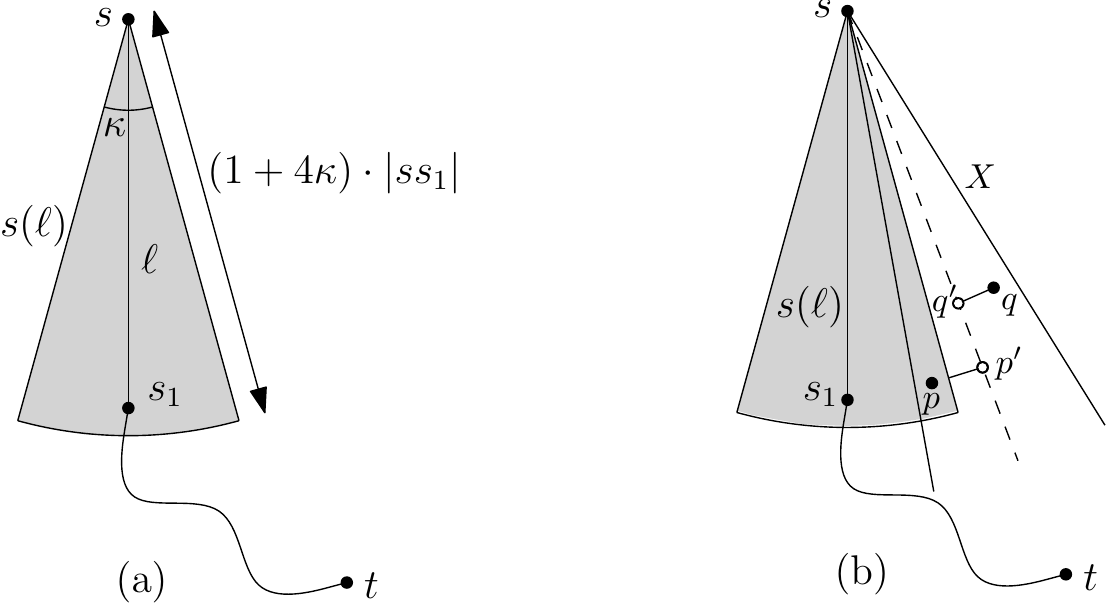}
      \caption{(a) Illustrating the definition of $s(\ell)$. (b) Illustrating the setting in Lemma~\ref{lem:apx_bound2}.} \label{fig:apx_bound1}
   \end{center}
 \end{figure}


\begin{lemma} \label{lem:apx_bound2}
 Let $\eps<1$ be positive constants. If $s(\ell)\cap \S_{\C} = \emptyset$ then:
    $$wt(\P) \leq wt(\P_2) \leq (1+\eps) \cdot wt(\P).$$
\end{lemma}
\ifproceedings
 \begin{proof} The proof can be found in Appendix~A. \hfill $\square$ \end{proof}
\else
\begin{proof}
 As above let $(s,s_1)$ be the first segment of $\P$, where $s_1$ lies on a road $c_1$. Assume w.l.o.g. that $c_1$ is belongs to the set of roads $\C_{ij}$ as defined in Section~\ref{ssec:constructing_candidates}(Type 2). Rotate the input such that $c_1$ is horizontal, below $s$ and going from right to left. Consider the construction of the $3$-tuples in Type~2, and let $[c'_1,s'_1,p]$ be the $3$-tuple reported when $\C_{ij}$ was processed in the direction towards $c_1$. See Fig.~\ref{fig:apx_bound2a}(b).

Consider the shortest path from $s$ to $t$ using the segment $(s,s'_1)$. We will have two cases depending on $s'_1$: either (1) $s_1$ and $s'_1$ lie on the same road, or (2) they lie on different roads.

Case 1: If $s'_1$ and $s_1$ both lie on road $c_1$ (with cost function $\alpha_1$) then we have:
\begin{eqnarray*}
 wt(\P_2) & \leq & |ss'_1| + \alpha_1 \cdot |s'_1,t_1| + \delta_G(t_1,t) \\
          & \leq & |ss'_1| + \alpha_1 \cdot |s'_1s_1| + \alpha_1 \cdot |s_1t_1| + \delta_G(t_1,t) \\
          & \leq & |ss'_1| + \alpha_1 \cdot |s'_1s_1| + \delta_G(s_1,t) \\
          & \leq & (1+4\kappa)\cdot |ss_1| + 5\kappa \cdot |ss_1| + \delta_G(s_1,t) \quad \mathrm{(since~}\theta \leq \kappa) \\
          & \leq & (1+9\kappa)\cdot |ss_1| + \delta_G(s_1,t) \\
          & \leq & (1+9\kappa)\cdot \delta_G(s,t) \\
          & < & (1+\eps) \cdot \delta_G(s,t)
\end{eqnarray*}
That completes the first part.

Part 2: If $s'_1$ and $s_1$ lie on different roads $c'_1$ and $c_1$, respectively, then $c'_1$ must lie between $s$ and $c_1$. This follows from the fact that $s(\ell)$ does not contain any endpoints and $c'_1$ is the first road hit. Furthermore, there must exist a an edge $(q',q)\in E$ such that $q'$ lies on $c'_1$ to the left of $s'_1$ and $q$ lies on $c_1$ to the right of $t_1$ and $\angle (u_1,q,q')=\arccos \alpha_1$.  See Fig.~\ref{fig:apx_bound2a}(b) for an illustration of case 2(a).

Consider the situation as depicted in Fig.~\ref{fig:apx_bound2a}(b). We will prove that the cost of the path from $s$ to $t_1$ via $s'_1, q'$ and $q$ is almost the same as the cost of the optimal path from $s$ to $t_1$ (that goes via $s_1)$. Recall that the cost function of $c_1$ and is $\alpha_1$ and the cost function of $c'_1$ is $\alpha'_1$. Furthermore, let $r$ be the intersection point between $c'_1$ and $(s,s_1)$.

Note that the cost of the path from $r$ to $q$ via $q'$ is maximized if $rq'$ forms an angle of $\theta$ with the horizontal line and $q'$ lies above $r$, as shown in Fig.~\ref{fig:apx_bound2a}(b). Furthermore, $|qq'|=|q'p|+|pq|$ and $|pq|=|rs_1|$.

\begin{eqnarray*}
 \alpha'_1 \cdot |rq'| + |q'p|  & \leq & \alpha'_1 \cdot |pr| \cdot \cos \theta + |pr| \cdot \sin \theta \\
         & < & |pr| (\alpha'_1 + \eps\cdot \alpha_{\min}/18) \\
         & < & \alpha'_1 \cdot |pr| (1 + \eps)
\end{eqnarray*}
Putting together the bounds we get:

\begin{eqnarray*}
 wt(\P_2) & \leq & |ss'_1| + \alpha'_1 \cdot |s'_1q'| + |q'q| + \delta_G(q,t) \\
          & \leq & |ss'_1| + \alpha'_1 \cdot |s'_1r| + \alpha'_1 \cdot |rq'| + (|q'p|+|pq|) +  \delta_G(q,t) \\
          & \leq & (1+9\kappa) \cdot |sr| + \alpha'_1 \cdot |rq'| + |q'p| +|rs_1| +  \delta_G(q,t) \quad \mathrm{(see~Part~1)} \\
          & \leq & (1+9\kappa) \cdot |sr| + |rs_1| + \alpha'_1 \cdot |pr| (1 + \eps) +  \delta_G(q,t) \\
             & < & (1+\eps) \cdot (|sr| + |rs_1| + \alpha'_1 \cdot |s_1q|) +  \delta_G(q,t)\\
             & < & (1+\eps) \cdot wt(\P)
\end{eqnarray*}
This completes the proof of Lemma~\ref{lem:apx_bound2}.
\hfill $\square$ \end{proof}

 \begin{figure} [h]
   \begin{center}
      \includegraphics[width=14cm]{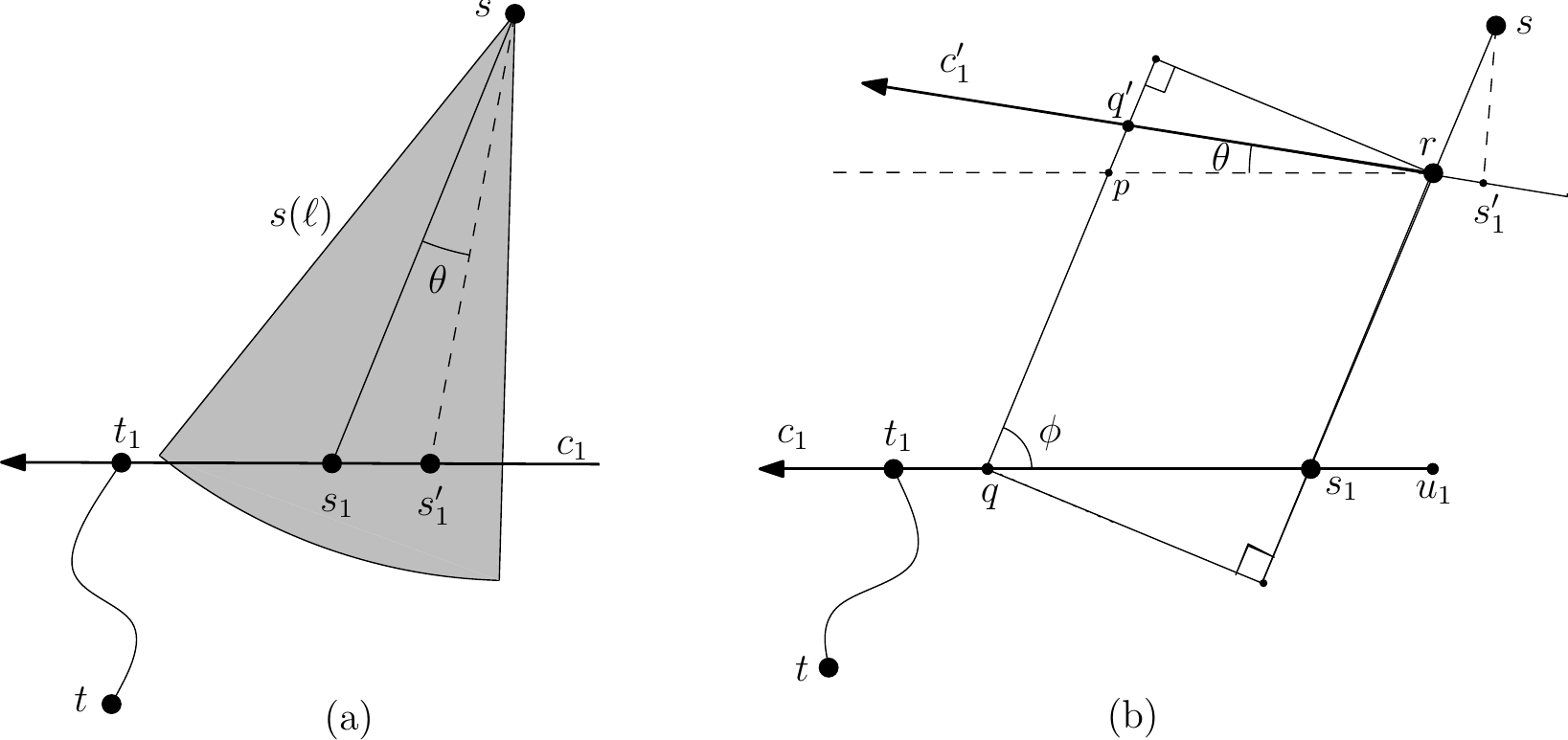}
      \caption{Illustrating (a) Case~1 and Case~2 (b) in the proof of Lemma~\ref{lem:apx_bound2}.} \label{fig:apx_bound2a}
   \end{center}
 \end{figure}
\fi

We can summarize this section (Lemmas ~\ref{lem:apx_bound1} and
~\ref{lem:apx_bound2}, Theorems ~\ref{thm:preprocessing} and
~\ref{thm:query}) with the following theorem:

\begin{theorem} \label{thm:total_time_s}
Given a transportation network $\T$ with $n$ roads in the
Euclidean plane, a destination point $t$ and a positive constant
$\eps$, one can preprocess $\T$ and $t$ in $O(n^2 \log n)$ time
and $O(n^2)$ space such that given a query point $s$, a
$(1+\eps)$-approximate quickest path can be calculated in
$O(\frac{1}{\alpha_{\min}\eps^2} \cdot \log n \log_{1+\eps}
\frac{\alpha_{\max}}{\alpha_{\min}})$ time.
\end{theorem}


\subsection{General case}~\label{sec:generalization}

In this section we turn our attention to the query version when we
are given two query points $s$ and $t$ in $\Reals^2$ and our goal
is to find a quickest path between $s$ and $t$ among $\C$. The
idea is the same as in the previous section. That is we perform
the exact same preprocessing steps as in the previous section
(omitting the destination point $t$), but with the exception that
$M$ contains all-pair shortest distances. Using Johnson's
algorithm~\cite{j-easps-77} the all-pairs shortest paths can be
computed in $O(n^4 \log n)$ using $O(n^4)$ space.

Given a query we compute the two candidate sets of type 1 and type 2 for both $s$ and $t$. For each pair of elements $p \in \D_1(s) \cup \D_2(s)$ and $q \in \D_1(t) \cup \D_2(t)$ compute a path between $s$ and $t$ as follows:
 \begin{description}
   \item[] if $p \in \D_1$ and $q\in \D_1$ then $|sp|+M[p,q]+|qt|$
   \item[] if $p \in \D_1$ and $q=[c, s_1, s_2]\in \D_2$ then $|sp|+M[p,s_2]+\alpha(c) \cdot |s_2s_1|+|s_1s|$
   \item[] if $p=[c, s_1, s_2] \in \D_2$ and $q\in \D_1$ then $|ss_1|+\alpha(c) \cdot |s_1s_2|+M[s_2,q]+|qt|$
   \item[] if $p=[c_1, s_1, s_2] \in \D_2$ and $q=[c_2, t_1, t_2]\in \D_2$ then $|ss_1|+\alpha(c_1) \cdot |s_1s_2|+M[s_2,t_2]+\alpha(c_2) \cdot |t_2t_1|+|t_1t|$
 \end{description}
Note that for a specific pair $p,q$ the calculated distance might not be a good approximation of the actual distance. However, for the shortest path there exists a pair $p,q$ such that the distance is a good approximation.

By putting together the results, we obtain the following theorem:

\begin{theorem} \label{thm:main_slow}
  Given a transportation network $\T$ with $n$ roads in the Euclidean plane and a positive constant $\eps$, one can preprocess $\T$ in $O(n^4\log n)$ time using $O(n^4)$ space such that given two query points $s$ and $t$ a $(1+\eps)$-approximate quickest path between $s$ and $t$ can be calculated in $O((\frac{1}{\alpha_{\min} \cdot \eps^2} \cdot \log_{1+\eps} \frac{\alpha_{\max}}{\alpha_{\min}})^2 \cdot \log n)$ time.
\end{theorem}


\subsection{Improving the complexity using the well-separated pair decomposition} \label{ssec:wspd}

The bottleneck of the preprocessing algorithm is the fact that $n^4$ shortest paths are computed in a graph of quadratic complexity. Is there a way to get around this? Since it suffices to approximate the shortest paths we can reduce the number of shortest path queries from $O(n^4)$ to $O(n^2)$, that is, linear in the number of vertices, using the well-separated pair-decomposition (WSPD).

\ifproceedings
 The definition of the WSPD and some basic properties can be found in Appendix~B.
\else
\begin{definition} [\cite{ck-dmpsa-95}]   \label{def:wellsep}
       Let $s>0$ be a real number, and let $A$ and $B$ be two finite
       sets of points in $\Reals^d$. We say that $A$ and $B$ are
       \emph{well-separated} with respect to $s$, if there are two
       disjoint $d$-dimensional balls $C_A$ and $C_B$, having the
       same radius, such that
       (i) $C_A$ contains the bounding box $R(A)$ of $A$,
       (i) $C_B$ contains the bounding box $R(B)$ of $B$, and
       (ii) the minimum distance between $C_A$ and $C_B$ is at least
       $s$ times the radius of $C_A$.
\end{definition}

The parameter $s$ will be referred to as the {\em separation constant}. The next lemma follows easily from Definition~\ref{def:wellsep}.

\begin{lemma} [\cite{ck-dmpsa-95}]  \label{lem:insamepair}
       Let $A$ and $B$ be two finite sets of points that are
       well-separated w.r.t.\ $s$, let $x$ and $p$ be points of $A$,
       and let $y$ and $q$ be points of $B$. Then
       (i) $|xy| \leq (1+4/s) \cdot |pq|$, and
       (ii) $|px| \leq (2/s) \cdot |pq|$.
\end{lemma}

\begin{definition}[\cite{ck-dmpsa-95}]  \label{def:WSPD}
       Let $S$ be a set of $n$ points in $\Reals^d$, and let $s>0$
       be a real number. A {\em well-separated pair decomposition}
       (WSPD) for $S$ with respect to $s$ is a sequence of pairs
       of non-empty subsets of~$S$,
       $(A_1,B_1) ,  \ldots , (A_m, B_m)$,
       such that
       \begin{enumerate}
       \item $A_i \cap B_i = \emptyset$, for all $i=1, \ldots, m$,
       \item for any two distinct points $p$ and $q$ of
             $S$, there is exactly one pair $(A_i,B_i)$
             in the sequence, such that
             (i) $p \in A_i$ and $q \in B_i$,
             or (ii) $q \in A_i$ and $p \in B_i$,
       \item $A_i$ and $B_i$ are well-separated w.r.t.\ $s$,
             for $1\leq i \leq m$.
       \end{enumerate}
       The integer $m$ is called the {\em size} of the WSPD.
\end{definition}

Callahan and Kosaraju showed that a WSPD of size $m = \O(s^dn)$ can be computed in $\O(s^dn+n \log n)$ time.
\fi

Construct the graph $G(V,E)$ of $\T$, as defined in Section~\ref{sec:Buildgraph}. Compute a WSPD $\{(A_i,B_i)\}_{i=1}^k$ of $V$ with respect to a separation constant $s=\frac {8} {\tau\cdot \alpha_{\min}}$, where $\tau<1$ is a positive constant given as input. Then for each well-separated pair $(A_i,B_i)$ pick two arbitrary points $a\in A_i$ and $b\in B_i$ as representative points, and calculate the shortest path in $G$ between $a$ and $b$. All paths are stored in a matrix $M'$. According to Definition~\ref{def:WSPD}, we have $O(n^2)$ well separated pairs. It follows that the number of the shortest path queries in $M$ is $O(n^2)$.

The queries are performed in almost the same way as above. The only difference is how the cost of the path between two points is calculated. Assume that we want the minimum transportation distance between two query points $p$ and $q$. According to Definition~\ref{def:WSPD} there exists a well-separated pair $\{A,B\}$ such that $p\in A$ and $q \in B$, or vice versa. Furthermore, let $r_A$ and $r_B$ be the representative points of $A$ and $B$, respectively. Instead of using the value in $M[p,q]$ we approximate the cost by $|pr_A|+M'[r_A,r_B]+|r_Bq|$.

\begin{theorem} \label{thm:wspd_bound}
  $$\delta_G(p,q) \leq  |pr_A|+M'[r_A,r_B]+|r_Bq| \leq (1+\tau) \cdot \delta_G(p,q).$$
\end{theorem}
\begin{proof}
  The left inequality is immediate, thus we will focus on the right inequality.
  \begin{eqnarray*}
     |pr_A|+M'[r_A,r_B]+|r_Bq| &  =   & |pr_A|+\delta_G(r_A,r_B) +|r_Bq|\\
                               & \leq & 2|pr_A|+2|r_Bq|+\delta_G(p,q) \\
                               & \leq & 4\cdot 2/s \cdot |pq| + \delta_G(p,q) \\
                               & \leq & \tau \cdot \alpha_{\min} \cdot |pq| + \delta_G(p,q) \\
                               & \leq & (1+\tau) \cdot \delta_G(p,q)
\end{eqnarray*}
Where the last step follows by using $|pq| \leq \delta_G(p,q)/\alpha_{\min}$.
\hfill $\square$ \end{proof}

Given a positive constant $\sigma<1$ we can obtain the claimed bounds by setting the constants appropriately (for example $\tau=\sigma/3$ and $\eps=\sigma/3)$. We get:

\begin{theorem} \label{thm:main_query_result}
  Given a transportation network $\T(S,\C)$ with $n$ roads and a positive constant $\eps$, one can preprocess $\C$ in $O((\frac {n}{\alpha_{\min}\cdot \eps})^2 \log n)$ time using $O((\frac {n}{\alpha_{\min} \cdot \eps})^2)$ space such that given two query points $s$ and $t$ a $(1+\eps)$-approximate quickest path between $s$ and $t$ can be calculated in $O((\frac{1}{\alpha_{\min} \cdot \eps^2} \cdot \log_{1+\eps} \frac{\alpha_{\max}}{\alpha_{\min}})^2 \cdot \log n)$ time.
\end{theorem}

Assuming $\alpha_{\min}$ and $\alpha_{\max}$ being constants the
bounds can be rewritten as: preprocessing is $O(n^2 \log n)$,
space is $O(n^2/\eps^2)$ and the query time is $O(1/\eps^4 \cdot
\log n)$.

\section{Concluding Remarks}
We considered the problem of computing a quickest path in a
transportation network. In the static case our algorithm has a
running time of  $O(n^2 \log n)$ which is a linear factor better
than the best algorithm known so far. We also introduce the query
version of the problem.

There are many open problems remaining. For example, can one develop a more efficient data structure that has a smaller dependency on $\alpha_{\min}$ and $\eps$? Also, is there a subquadratic time algorithm for the static case?  If not, can we prove that the problem is 3sum-hard? Bae and Chwa~\cite{GR-06} proved that the shortest path map for a transportation metric can have $\Omega(n^2)$ size, which indicates that the problem may indeed be 3sum-hard.


\ifproceedings
  \newpage
  \appendix
  \section{Appendix: Proofs}

\noindent {\bf Property 3:}
 There exists an optimal path $\P'$ of total cost $wt(\P)$ with $\C_{\P'}=\C_{\P}$ that fulfills Properties~\ref{Property1}-\ref{Property2} such that for any two consecutive roads $c_i$ and $c_{i+1}$ in $\C_{\P'}$ the straight-line segment $(t_i,s_{i+1})$ of $\P'$ must have an endpoint at an endpoint of $c_i$ or $c_{i+1}$, respectively.
\begin{proof}
Assume the opposite, i.e., $(t_i,s_{i+1})$ does not coincide with
any of the endpoints of $c_i$ or $c_{i+1}$. Consider the three
segment path from $s_i$ to $t_{i+1}$, that is, $(s_i,t_i)$,
$(t_i,s_{i+1})$ and $(s_{i+1},t_{i+1})$. The length of this path
is:
$$\alpha_i\cdot |s_it_i|+ |t_is_{i+1}|+\alpha_{i+1}\cdot |s_{i+1}t_{i+1}|.$$
According to Property ~\ref{Property2} the orientation of
$(t_i,s_{i+1})$ is fixed, which implies that the weight of the
path is a linear function only depending on the position of $t_i$
(or $s_{i+1}$). Hence, moving $t_i$ in one direction will
monotonically increase the weight of the path until one of two
cases occur: (1) either $t_i$ or $s_{i+1}$ encounters an endpoint
of $c_i$ or $c_{i+1}$, or (2) $t_i=s_i$ or $s_{i+1}=t_{i+1}$. If
(1) then we have a contradiction since we assumed $(t_i,s_{i+1})$
did not coincide with any endpoint. And if (2) then we have a
contradiction since $\P$ must follow both $c_i$ and $c_{i+1}$
(again from the definition of $\C_{\P}$). \hfill $\square$
\end{proof}


\vspace{1cm}
\noindent {\bf Lemma 3:}
 Given a point set $\S_{\C}$ and a positive constant $\eps$ one can preprocess $\S_{\C}$ into a data structure of size $O(n/\eps)$ in $O(1/\eps \cdot n \log n)$ time such that given a query point $s$ the point set $\D_1$, of size at most $36\pi/\eps$, can be reported in $O(1/\eps \cdot \log n)$ time.
\begin{proof}
  Given a direction $d$ and a point $s$ let $\ell_d(s)$ denote the infinite ray originating at $s$ with direction~$d$, see Fig.~\ref{fig:Query1}b. Let $C(s,d,\theta)$ be the cone with apex at $s$, bisector $\ell_d(s)$ and angle $\theta$. It has been shown (see for example Section~4.1.2 in~\cite{ns-gsn-07} or Lemma~2 in~\cite{bgm-otg-04}) that $\S_{\C}$ can be preprocessed in $O(n \log n)$ time into a data structure of size $O(n)$ such that given a query point $s$ in the plane the data structure returns the point in $C(s,d,\theta)$ whose orthogonal projection onto $\ell_d(s)$ is closest to $s$ in $O(\log n)$ time. We have $36\pi/\eps$ directions, thus the lemma follows. \hfill $\square$
\end{proof}


\vspace{1cm}
\noindent {\bf Lemma 4:}
 Given a set of roads $\C$ and a positive constant $\eps$ one can preprocess $\C$ in $O(n \log n)$ time into a data structure of size $O(n)$ such that given a query point $s$ the set $\D_2$ can be reported in $O(\frac {1}{\alpha_{\min}\cdot\eps^2} \cdot \log n \log_{1+\eps} (\alpha_{\max}/\alpha_{\min}))$ time. The number of $3$-tuples in $\D_2$ is $O(\frac{1} {\alpha_{\min} \eps^2} \cdot \log_{1+\eps} \frac{\alpha_{\max}}{\alpha_{\min}})$.
\begin{proof}
The preprocessing consists of two steps: (1) partitioning $\C$ into the sets $\C_{i,j}$, $1\leq i \leq m$ and $1\leq j \leq b$, and (2) preprocessing each set $\C_{i,j}$ into a data structure that answers ray shooting queries efficiently.

The first part is easily done in $O(n \log n)$ time by sorting the
roads first with respect to their orientation and then with
respect to their weight.

The second step of the preprocessing can be done by building two trapezoidal maps $T_{up}(\C_{i,j})$ and $T_{down}(\C_{i,j})$ (also know as a vertical decomposition) of each set $\C_{i,j}$ as follows (see Fig.~\ref{fig:Query2}c). Rotate $\C_{i,j}$ such that $\gamma_{up}(i,j)$ is vertical and upward. Build a trapezoidal map $T_{up}(\C_{i,j})$ of $\C_{i,j}$ as described in Chapter~6.1 in~\cite{bcko-cgaa-08}. Then preprocess $T_{up}(\C_{i,j})$ to allow for planar point location. Note that every face in $T(\C)$ either is a triangle or a trapezoid, and the left and right edges of each face (if they exists) are vertical. The trapezoidal map $T_{down}(\C_{i,j})$ can be computed in the same way by rotating $\C_{i,j}$ such that $\gamma_{down}(i,j)$ is vertical and upward. The total time needed for this step is $O(n \log n)$ and it requires $O(n)$ space.

When a query point $s$ is given, two ray shooting queries are
performed for each set $\C_{i,j}$. However, instead we perform a
point location in the trapezoidal maps $T_{up}(\C_{i,j})$ and
$T_{down}(\C_{i,j})$. Consider $T_{up}(\C_{i,j})$ and let $f$ be
the face in the map containing $s$. The top edge of $f$
corresponds to the first road $c_{up}$ hit by a ray emanating from
$s$ in direction $\gamma_{up}(i,j)$. When $c_{up}$ is found we
just add to $\D_2$ the first vertex on $c_{up}$ in $G$ to the
right of $s$. The same process is repeated for
$T_{down}(\C_{i,j})$. Performing the point location requires
$O(\log n)$ time per trapezoidal map, thus the total query time is
$O(mb \log n) = O(\frac {1}{\alpha_{\min}\cdot\eps^2} \cdot \log n
\log_{1+\eps} (\alpha_{\max}/\alpha_{\min}))$. \hfill $\square$
\end{proof}


\vspace{1cm}
\noindent {\bf Theorem 4:}
 A query can be answered in time $O(\frac{1}{\alpha_{\min}\eps^2} \cdot \log n \log_{1+\eps} (\frac{\alpha_{\max}}{\alpha_{\min}}))$.
\begin{proof}
 As above we divide the analysis into two parts: type 1 and type 2.
 \begin{description}
  \item[] Type 1: According to Lemma~\ref{lem:type1_set} the number of type 1 candidate points is at most $32\pi/\eps$ and can be computed in time $O(1/\eps\cdot \log n)$. Computing the cost of a quickest path for a point in $\D_1$ can be done in constant time. Thus, the query time is $O(1/\eps \log n)$.
  \item[] Type 2: According to Lemma~\ref{lem:type2_set} the number of $3$-tuples in $\D_2$ is $O(\frac{1}{\alpha_{\min}\eps^2} \cdot \log_{1+\eps} \alpha_{\max}/\alpha_{\min})$ and can be computed in time $O(\frac{1}{\alpha_{\min}\eps^2} \cdot \log n \log_{1+\eps} \alpha_{\max}/\alpha_{\min})$. Each element in $\D_2$ is then processed in $O(1)$ time, thus $O(\frac{1}{\alpha_{\min}\eps^2} \cdot \log n \log_{1+\eps} \alpha_{\max}/\alpha_{\min})$ in total.
 \end{description}
 Summing up we get the bound stated in the theorem.
\hfill $\square$ \end{proof}


\vspace{1cm}
\noindent {\bf Lemma 6:}
 Let $\eps<1$ be positive constants. If $s(\ell)\cap \S_{\C} = \emptyset$ then:
    $$wt(\P) \leq wt(\P_2) \leq (1+\eps) \cdot wt(\P).$$
\begin{proof}
As above let $(s,s_1)$ be the first segment of $\P$, where $s_1$
lies on a road $c_1$. Assume w.l.o.g. that $c_1$ belongs to the
set of roads $\C_{ij}$ as defined in
Section~\ref{ssec:constructing_candidates}(Type 2). Rotate the
input such that $c_1$ is horizontal, below $s$ and going from
right to left. Consider the construction of the $3$-tuples in Type
2, and let $[c'_1,s'_1,p]$ be the $3$-tuple reported when
$\C_{ij}$ was processed in the direction towards $c_1$. See
Fig.~\ref{fig:apx_bound2a}(b).

Consider the shortest path from $s$ to $t$ using the segment
$(s,s'_1)$. We will have two cases depending on $s'_1$: either (1)
$s_1$ and $s'_1$ lie on the same road, or (2) they lie on
different roads.

Case 1: If $s'_1$ and $s_1$ both lie on road $c_1$ (with cost
function $\alpha_1$) then we have:
\begin{eqnarray*}
 wt(\P_2) & \leq & |ss'_1| + \alpha_1 \cdot |s'_1t_1| + \delta_G(t_1,t) \\
          & \leq & |ss'_1| + \alpha_1 \cdot |s'_1s_1| + \alpha_1 \cdot |s_1t_1| + \delta_G(t_1,t) \\
          & \leq & |ss'_1| + \alpha_1 \cdot |s'_1s_1| + \delta_G(s_1,t) \\
          & \leq & (1+4\kappa)\cdot |ss_1| + 5\kappa \cdot |ss_1| + \delta_G(s_1,t) \quad \mathrm{(since~}\theta \leq \kappa) \\
          & \leq & (1+9\kappa)\cdot |ss_1| + \delta_G(s_1,t) \\
          & \leq & (1+9\kappa)\cdot \delta_G(s,t) \\
          & < & (1+\eps) \cdot \delta_G(s,t)
\end{eqnarray*}
That completes the first case.

Case 2: If $s'_1$ and $s_1$ lie on different roads $c'_1$ and
$c_1$, respectively, then $c'_1$ must lie between $s$ and $c_1$.
This follows from the fact that $s(\ell)$ does not contain any
endpoints and $c'_1$ is the first road hit. Furthermore, there
must exist an edge $(q',q)\in E$ such that $q'$ lies on $c'_1$ to
the left of $s'_1$ and $q$ lies on $c_1$ to the right of $t_1$ and
$\angle (u_1,q,q')=\arccos \alpha_1$.  See
Fig.~\ref{fig:apx_bound2a}(b) for an illustration of case 2(a).
This is easily argued by a proof of contradiction and we leave it
for the full version of the paper.

Consider the situation as depicted in Fig.~\ref{fig:apx_bound2a}(b). We will prove that the cost of the path from $s$ to $t_1$ via $s'_1, q'$ and $q$ is almost the same as the cost of the optimal path from $s$ to $t_1$ (that goes via $s_1)$. Recall that the cost function of $c_1$ and is $\alpha_1$ and the cost function of $c'_1$ is $\alpha'_1$. Furthermore, let $r$ be the intersection point between $c'_1$ and $(s,s_1)$.

Note that the cost of the path from $r$ to $q$ via $q'$ is maximized if $(r,q')$ forms an angle of $\theta$ with the horizontal line and $q'$ lies above $r$, as shown in Fig.~\ref{fig:apx_bound2a}(b). Furthermore, $|qq'|=|q'p|+|pq|$ and $|pq|=|rs_1|$.

\begin{eqnarray*}
 \alpha'_1 \cdot |rq'| + |q'p|  & \leq & \alpha'_1 \cdot |pr| \cdot \cos \theta + |pr| \cdot \sin \theta \\
         & < & |pr| (\alpha'_1 + \eps\cdot \alpha_{\min}/18) \\
         & < & \alpha'_1 \cdot |pr| (1 + \eps)
\end{eqnarray*}
Putting together the bounds we get:

\begin{eqnarray*}
 wt(\P_2) & \leq & |ss'_1| + \alpha'_1 \cdot |s'_1q'| + |q'q| + \delta_G(q,t) \\
          & \leq & |ss'_1| + \alpha'_1 \cdot |s'_1r| + \alpha'_1 \cdot |rq'| + (|q'p|+|pq|) +  \delta_G(q,t) \\
          & \leq & (1+9\kappa) \cdot |sr| + \alpha'_1 \cdot |rq'| + |q'p| +|rs_1| +  \delta_G(q,t) \quad \mathrm{(see~Part~1)} \\
          & \leq & (1+9\kappa) \cdot |sr| + |rs_1| + \alpha'_1 \cdot |pr| (1 + \eps) +  \delta_G(q,t) \\
             & < & (1+\eps) \cdot (|sr| + |rs_1| + \alpha'_1 \cdot |s_1q|) +  \delta_G(q,t)\\
             & < & (1+\eps) \cdot wt(\P)
\end{eqnarray*}
This completes the proof of Lemma~\ref{lem:apx_bound2}.
\hfill $\square$ \end{proof}

 \begin{figure} [h]
   \begin{center}
      \includegraphics[width=14cm]{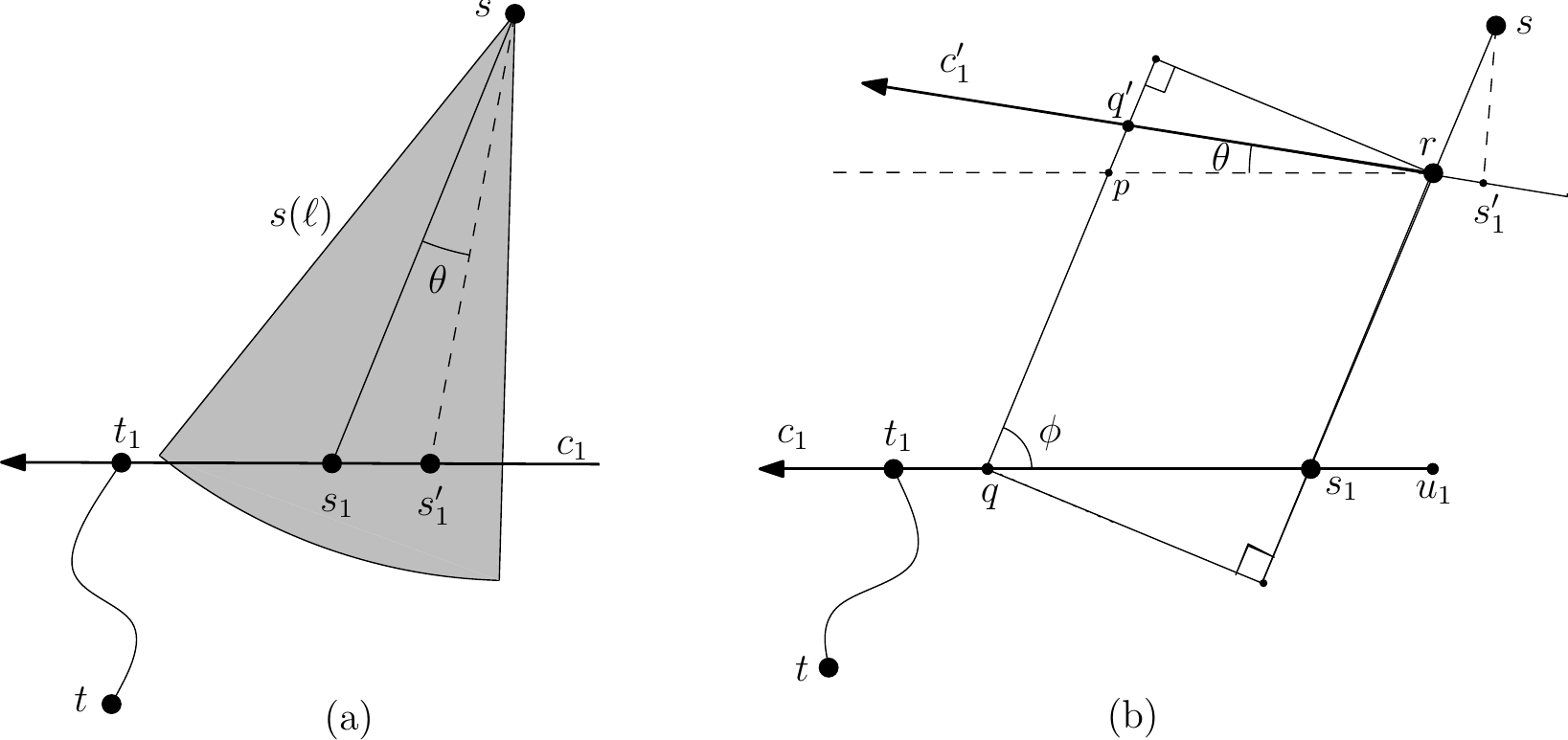}
      \caption{Illustrating (a) Case~1 and Case~2 (b) in the proof of Lemma~\ref{lem:apx_bound2}.} \label{fig:apx_bound2a}
   \end{center}
 \end{figure}


\section{Appendix: WSPD}

\begin{definition} [\cite{ck-dmpsa-95}]   \label{def:wellsep}
       Let $s>0$ be a real number, and let $A$ and $B$ be two finite
       sets of points in $\Reals^d$. We say that $A$ and $B$ are
       \emph{well-separated} with respect to $s$, if there are two
       disjoint $d$-dimensional balls $C_A$ and $C_B$, having the
       same radius, such that
       (i) $C_A$ contains the bounding box $R(A)$ of $A$,
       (i) $C_B$ contains the bounding box $R(B)$ of $B$, and
       (ii) the minimum distance between $C_A$ and $C_B$ is at least
       $s$ times the radius of $C_A$.
\end{definition}

The parameter $s$ will be referred to as the {\em
separation constant}. The next lemma follows easily from
Definition~\ref{def:wellsep}.

\begin{lemma} [\cite{ck-dmpsa-95}]  \label{lem:insamepair}
       Let $A$ and $B$ be two finite sets of points that are
       well-separated w.r.t.\ $s$, let $x$ and $p$ be points of $A$,
       and let $y$ and $q$ be points of $B$. Then
       (i) $|xy| \leq (1+4/s) \cdot |pq|$, and
       (ii) $|px| \leq (2/s) \cdot |pq|$.
\end{lemma}

\begin{definition}[\cite{ck-dmpsa-95}]  \label{def:WSPD}
       Let $S$ be a set of $n$ points in $\Reals^d$, and let $s>0$
       be a real number. A {\em well-separated pair decomposition}
       (WSPD) for $S$ with respect to $s$ is a sequence of pairs
       of non-empty subsets of~$S$,
       $(A_1,B_1) ,  \ldots , (A_m, B_m)$,
       such that
       \begin{enumerate}
       \item $A_i \cap B_i = \emptyset$, for all $i=1, \ldots, m$,
       \item for any two distinct points $p$ and $q$ of
             $S$, there is exactly one pair $(A_i,B_i)$
             in the sequence, such that
             (i) $p \in A_i$ and $q \in B_i$,
             or (ii) $q \in A_i$ and $p \in B_i$,
       \item $A_i$ and $B_i$ are well-separated w.r.t.\ $s$,
             for $1\leq i \leq m$.
       \end{enumerate}
       The integer $m$ is called the {\em size} of the WSPD.
\end{definition}

Callahan and Kosaraju showed that a WSPD of size $m = \O(s^dn)$ can be computed in $\O(s^dn+n \log n)$ time.
\fi


\begin{thebibliography}{99}
\bibliographystyle{plain}

\bibitem{ahiklmps-pptmi-01}
M. Abellanas, F. Hurtado, C. Icking, R. Klein, E. Langetepe, L.
Ma, B. Palop del R\'{\i}o and V. S\'{a}cristan.
\newblock Proximity problems for time metrics induced by the L1
metric and isothetic networks.
\newblock In Actas de los IX Encuentros de Geometr´ýa Computacional,
pages 175–182, Girona, 2001.

\bibitem{ahiklmprs-vdsnh-03}
M. Abellanas, F. Hurtado, C. Icking, R. Klein, E. Langetepe, L.
Ma, B. Palop del R\'{\i}o and V. S\'{a}cristan.
\newblock Voronoi diagram for services neighboring a highway.
\newblock Information Processessing Letters, 86:283–-288, 2003.


\bibitem{as-anspt-93}
  P. K. Agarwal and M. Sharir.
  \newblock Applications of a new space partitioning technique.
  \newblock Discrete \& Computational Geometry, 9:11--38, 1993.

\bibitem{aap-qpssc-04}
O. Aichholzer, F. Aurenhammer and B. Palop del R\'{\i}o.
\newblock Quickest paths, straight skeletons, and the city Voronoi diagram.
\newblock Discrete \& Computatational Geometry, 31:17–-35, 2004.

\bibitem{bkc-occvd-09}
S. W. Bae, J.-H. Kim and K.-Y. Chwa.
\newblock Optimal construction of the city Voronoi diagram.
\newblock International Journal on Computational Geometry \&
Applications, 19(2):95--117, 2009.

\bibitem{GR-06}
S. W. Bae and K.-Y. Chwa.
\newblock Voronoi Diagrams for a Transportation Network on the
Euclidean Plane.
\newblock International Journal on Computational Geometry \&
Applications, 16:117--144, 2006.

\bibitem{SPGR-05}
S. W. Bae and K.-Y. Chwa.
\newblock Shortest Paths and Voronoi Diagrams with Transportation
Networks Under General Distances.
\newblock ISAAC, pages 1007--1018, 2005.

\bibitem{bcko-cgaa-08}
M. de Berg, O. Cheong, M. van Kreveld and M. Overmars.
\newblock Computational Geometry: Algorithms and Applications. (3rd edition).
\newblock Springer-Verlag, Heidelberg, 2008.

\bibitem{b-mddc-78}
J. L. Bentley.
\newblock Multidimensional divide-and-conquer.
\newblock Communications of the ACM, 23(5):214--228, 1978.

\bibitem{bgm-otg-04}
P. Bose, J. Gudmundsson and P. Morin.
\newblock Ordered theta graphs.
\newblock Computational geometry -- Theory \& Applications,
28(1):11--18, 2004.

\bibitem{ck-dmpsa-95}
P.~B. Callahan and S.~R. Kosaraju.
\newblock A decomposition of multidimensional point sets with applications to  $k$-nearest-neighbors and $n$-body potential fields.
\newblock {\em Journal of the {ACM}}, 42:67--90, 1995.

\bibitem{cj-arsis-92}
 S. W. Cheng and R. Janardan.
 \newblock Algorithms for ray-shooting and intersection searching.
 \newblock Journal of Algorithms 13:670--692, 1992.


\bibitem{d-ntpcg-59}
E.~W. Dijkstra.
\newblock A note on two problems in connexion with graphs.
\newblock {\em Numerische Mathematik}, 1:269–-271, 1959.

\bibitem{gsw-ccvdf-08}
R. G\"{o}rke, C.-C. Shin and A. Wolff.
\newblock Constructing the City Voronoi Diagram Faster.
\newblock International Journal of Computational Geometry and
Applications 18(4):275--294, 2008.

\bibitem{hs-parss-95}
  J. Hershberger and S. Suri.
  \newblock A pedestrian approach to ray shooting: Shoot a ray, take a walk.
  \newblock Journal of Algorithms 18:403--431, 1995.

\bibitem{j-easps-77}
  D.~B.~Johnson
  \newblock Efficient algorithms for shortest paths in sparse networks.
  \newblock Journal of the ACM, 24(1):1-–13, 1977.

\bibitem{ns-gsn-07}
G.~Narasimhan and M.~Smid
\newblock Geometric Spanner Networks.
\newblock Cambridge University Press, 2007.

\end{thebibliography}
\end{document}